\documentclass[final,11pt,onecolumn]{IEEEtran}

\usepackage{amssymb}
\usepackage{amsmath}
\usepackage{epsfig}

\newtheorem{theorem}{Theorem}

\begin{document}
\bibliographystyle{ieeetr}

\title{On Degrees of Freedom Region of MIMO Networks without CSIT}
\author{\authorblockN{Chiachi Huang, Syed A. Jafar, Shlomo Shamai (Shitz), Sriram Vishwanath}
\thanks{C. Huang was with Department of Electrical Engineering and Computer Science,
University of California, Irvine, California, USA while this work was done.
He is now with Department of Communications Engineering, Yuan Ze University, Chongli, Taiwan
({\small e-mail: chiachi@saturn.yzu.edu.tw}).
S. A. Jafar is with Department of Electrical Engineering and Computer Science,
University of California, Irvine, California, USA ({\small e-mail: syed@uci.edu}).
S. Shamai is with Department of Electrical Engineering, Technion-Israel Institute of Technology,
Technion City, Haifa, Israel
({\small e-mail: sshlomo@ee.technion.ac.il}).
The work of S. Shamai was supported by the European Commission in
the framework of the FP7 Network of Excellence in Wireless Communications
NEWCOM++.
S. Vishwanath is with Department of Electrical Engineering, University of Texas, Austin, Texas, USA
({\small e-mail: sriram@ece.utexas.edu}).
}}
\maketitle

\begin{abstract}
In this paper, we study the effect of the absence of channel knowledge for multiple-input-multiple-output
(MIMO) networks. Specifically, we assume perfect channel state information at the receivers, no channel
state information at the transmitter(s), and independent identically distributed (i.i.d.) Rayleigh fading across antennas, users and time slots. We provide the characterization of the degrees of freedom (DoF)
region for a 2-user MIMO broadcast channel. We then provide a DoF region outer bound for a 2-user MIMO
interference channel. This bound is shown to be tight for all possible combinations of the number of
antennas at each node except for one case. As a byproduct of this analysis we point out the potential of interference alignment in the 2-user MIMO interference channel with no CSIT.
\end{abstract}


\newpage
\section{Introduction}
Multiple-input-multiple-output (MIMO) systems are capable of providing remarkably
higher capacity compared to traditional single-input-single-output (SISO) systems. One of the key features of
MIMO systems is the possibility of multiplexing signals in space. The ability of multiplexing signals
in space is measured by spatial multiplexing gain \cite{Zheng_Tse}, which is also called capacity prelog or
degrees of freedom (DoF). For point-to-point MIMO communication systems, it has been shown that the
availability of the channel state information at the transmitter (CSIT) does not affect the spatial
multiplexing gain \cite{Telatar}.
Unlike the point-to-point case, in a network with distributed processing units, references
\cite{Caire_Shamai_CSI, Caire_Shamai_ITrans, Jafar_scalar, Jafar_mobile}, and \cite{Lapidoth} show that in the absence of channel knowledge, spatial
multiplexing gain may be lost. For example, the DoF of the fading multiple-input-single-output (MISO) broadcast
channel with $M$ antennas at the transmitter and $1$ antenna at all $M$ receivers is $M$ when perfect
channel state information is available at the transmitter (perfect CSIT)
\cite{Yu_Cioffi,Viswanath_Tse_BC,Vishwanath_Jindal_Goldsmith}. However, the DoF of the same system is
only $1$ when channel state information is not available at the transmitter (no CSIT) \cite{Jafar_scalar}.
Further understanding about the availability of channel knowledge and its effect on the degrees freedom
of the networks can provide insights into the design and optimization of the wireless networks.

In this paper, we first study the DoF region of a 2-user MIMO broadcast channel where the transmitter
is equipped with $M$ antennas and receivers $1,2$ are equipped with $N_1,N_2$ antennas, respectively,
under the assumption of perfect channel state information at the receivers (perfect CSIR) and no CSIT. The channel model assumed throughout this paper is Rayleigh fading that is i.i.d. across antennas, users and time slots.
An exact characterization of the DoF region is given,
and our result shows that a simple time division scheme between the two users is DoF-region optimal.
We also extend the outer bound of the DoF region to find the capacity region for a specific 2-user
MIMO broadcast channel.
Under the same assumptions of perfect CSIR and no CSIT,
we then generalize the outer bound derived for the 2-user MIMO broadcast channel to provide a
DoF-region outer bound for a 2-user MIMO interference channel where transmitters 1 and 2 are
equipped with $M_1, M_2$ antennas respectively and receivers 1 and 2 are equipped with
$N_1, N_2$ antennas respectively.
The outer bound is found to be tight in all cases except for one special
class of channels where the DoF region remains unknown. The loss of DoF
due to lack of CSIT depends on the relative magnitudes of
$M_1, M_2, N_1, N_2$, and is found to be more severe when transmitters
carry more antennas than receivers and less severe when receivers carry
more antennas than receivers. For example, if transmitters have more antennas than their corresponding receivers, $M_1\geq N_1$ and $M_2\geq N_2$
then the benefits of multiplexing are lost to the extent that time-division between the two users is DoF-region optimal. However, if the receivers have more antennas than their \emph{interfering} transmitters, $N_1\geq M_2$ and $N_2\geq M_1$, then there is no loss in DoF relative to the case where perfect CSIT is available.  For the class of channels where the DoF-region remains open, the possibility of interference alignment is seen to be the principal challenge.


The rest of the paper is organized as follows. Section \ref{sec:systemmodel} describes the models.
Section \ref{sec:rslts} summarizes our main results.
In Section \ref{sec:prf}, we prove all the theorems given in Section \ref{sec:rslts}.
Section \ref{sec:cpcty_bc} studies the capacity region of a specific 2-user MIMO broadcast channel.
Section \ref{sec:cnlsn} concludes the paper.

Regarding notation, we use ${\bf V}_i$ to denote the $i^{th}$ element of a vector ${\bf V},$
and ${\bf V}_{i:j}$ for the elements that start from element $i$ and end at element $j.$
We define ${\bf V}_{i:j} = \phi$ if $i > j.$
We say $f(x)=o(g(x))$ if $\lim_{x \rightarrow \infty} \frac{f(x)}{g(x)} = 0.$

\section{System Model}
\label{sec:systemmodel}
\subsection{2-User MIMO Broadcast channel}
\label{sec:bcmodel}
Consider the $2$-user Gaussian MIMO broadcast channel where the transmitter is equipped with $M$ antennas
and receivers $1,2$ are equipped with $N_1,N_2$ antennas, respectively. The channel is described by the
input-output relationship:
\begin{eqnarray}
{\bf Y}^{[1]}(t)&=&{\bf H}^{[1]}(t){\bf X}(t)+{\bf Z}^{[1]}(t)\\
{\bf Y}^{[2]}(t)&=&{\bf H}^{[2]}(t){\bf X}(t)+{\bf Z}^{[2]}(t)
\end{eqnarray}
where at the $t^{th}$ channel use, ${\bf Y}^{[i]}(t)$, ${\bf Z}^{[i]}(t)$ are the $N_i\times 1$ vectors
representing the channel output and additive white Gaussian noise at receiver $i$, ${\bf H}^{[i]}(t)$ is
the $N_i\times M$ channel matrix corresponding to receiver $i$, $i\in\{1,2\}$, and ${\bf X}(t)$ is the
$M\times 1$ input vector. The elements of ${\bf H}^{[i]}(t)$ and ${\bf Z}^{[i]}(t)$, $i=1,2$, are
independent and identically distributed (i.i.d.) (both across space and time) circularly symmetric complex Gaussian random variables with zero
mean and unit variance. We assume perfect CSIR, i.e., each receiver has perfect knowledge of all channel
matrices at each instant, and no CSIT, i.e., the transmitter has no knowledge of the instantaneous values
taken by the channel coefficients. To avoid cumbersome notation, we will henceforth suppress the channel
use index $t$ where it does not cause ambiguity.

The transmit power constraint is expressed as:
\begin{eqnarray}
\mbox{E}[||{\bf X}||^2]\leq P.
\end{eqnarray}
There are two independent messages $W_1,W_2$, associated with rates $R_1,R_2$, to be communicated from the transmitter to receivers $1,2,$ respectively. The capacity region $C(P)$ is the set of all rate pairs $(R_1, R_2)$ for which the probability of error can be driven arbitrarily close to zero by using suitably long codewords. The degrees of freedom region is defined as follows:
\begin{equation}
\mathbf{D}\triangleq \left\{(d_1,d_2)\in\mathbb{R}_2^+: \exists (R_1(P), R_2(P))\in C(P) ~~ \mbox{ s.t.} ~~ d_i=\lim_{P\rightarrow\infty}\frac{R_i(P)}{\log(P)}, ~~ i=1,2.\right\}.
\end{equation}

\subsection{2-user MIMO Interference Channel}
\label{sec:icmodel}
Consider the $2$-user Gaussian MIMO interference channel where transmitters $1,2$ are equipped with
$M_1,M_2$ antennas, respectively, and receivers $1,2$ are equipped with $N_1,N_2$ antennas, respectively.
Without loss of generality, we assume that $N_1 \leq N_2.$ The channel is described by the input-output
relationship:
\begin{eqnarray}
{\bf Y}^{[1]}(t)
& = &
{\bf H}^{[11]}(t){\bf X}^{[1]}(t)+{\bf H}^{[12]}(t){\bf X}^{[2]}(t)+{\bf Z}^{[1]}(t)\\
{\bf Y}^{[2]}(t)
& = &
{\bf H}^{[21]}(t){\bf X}^{[1]}(t)+{\bf H}^{[22]}(t){\bf X}^{[2]}(t)+{\bf Z}^{[2]}(t)
\end{eqnarray}
where at the $t^{th}$ channel use, ${\bf Y}^{[j]}(t)$, ${\bf Z}^{[j]}(t)$ are the $N_j\times 1$ vectors
representing the channel output and additive white Gaussian noise at receiver $j$, ${\bf H}^{[ji]}(t)$ is
the $N_j\times M_i$ channel matrix corresponding to receiver $j$, and ${\bf X}^{[i]}(t)$ is the
$M_i\times 1$ input vector, $i,j\in\{1,2\}$. The following assumptions are similar to those in Section
\ref{sec:systemmodel}. The elements of ${\bf H}^{[ji]}(t)$ and ${\bf Z}^{[j]}(t)$, $i,j\in\{1,2\}$,
are i.i.d. (across space and time) circularly symmetric complex Gaussian random variables with
zero mean and unit variance. We assume perfect CSIR and no CSIT.

The transmit power constraint is expressed as:
\begin{eqnarray}
\mbox{E}[||{\bf X}^{[i]}||^2]\leq P, ~~~i=1,2.
\end{eqnarray}
There are two independent messages $W_1,W_2$, associated with rates $R_1,R_2$, to be communicated from the transmitter $1$ to receiver $1$ and from the transmitter $2$ to receiver $2$, respectively. The standard definitions of the capacity region and the DoF region are the same with those in the previous section.

\section{Main Results}
\label{sec:rslts}
We provide our main results in this section. We start with the broadcast channel.

\subsection{2-User MIMO Broadcast Channel}
\label{sec:bc_rslts}
\begin{theorem}
\label{thm:dofbc}
The DoF region of the MIMO broadcast channel with no CSIT, as defined in Section \ref{sec:bcmodel}, is
the following:
\begin{equation}
{\bf D} =
\left\{ (d_1,d_2)\in\mathbb{R}_2^+: \frac{d_1}{\min({M,N_1})}+\frac{d_2}{\min({M,N_2})}\leq 1 \right\}.
\end{equation}
\end{theorem}
The proof of this theorem is provided in Section \ref{sec:prf}.

\begin{figure}
\begin{center}
\includegraphics[width=400pt, trim=0 160 0 120, clip]{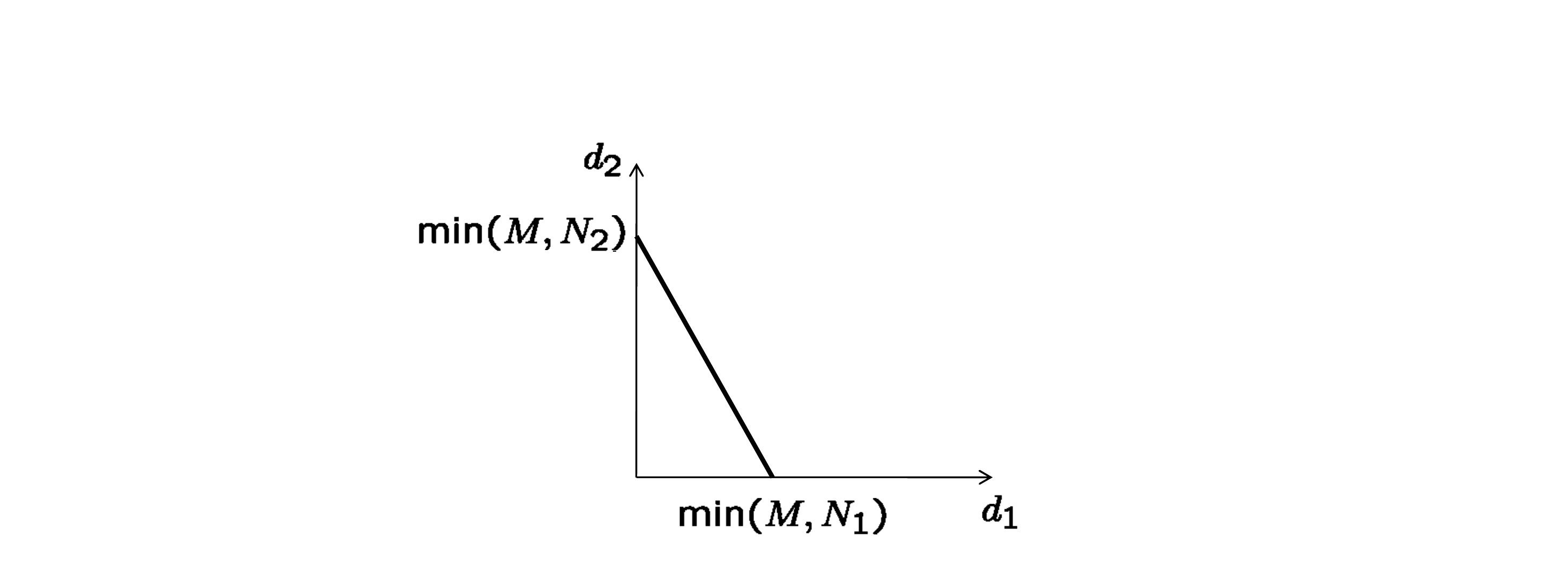}
\caption{The DoF region of the 2-user MIMO broadcast channel.}
\end{center}
\end{figure}

{\it Remark 1:} It's easy to see that a simple time division multiplexing between the two users is
optimal from the perspective of degrees of freedom.
Note that time division multiplexing between the two users is DoF-region optimal for a 2-user MIMO broadcast
channel with perfect CSIT only when $M \leq N_1$ and $M \leq N_2.$ Thus, the absence of CSIT does not affect the DoF
region of the channel only when $M \leq N_1$ and $M \leq N_2.$

{\it Remark 2:} Note that in all prior DoF region characterizations, including MIMO multiple access channel,
MIMO broadcast channel, and MIMO interference channel (all with perfect CSIT), the slanting face has a
slope of 45 degrees, i.e., a degree of freedom lost by one user results in a degree of freedom
gained by another user. In other words, the tradeoff is on equal terms. However, with no CSIT, the
slope of the slanting face is not 45 degrees, and the tradeoff is no longer even. For example, suppose
$N_1 \leq N_2.$ Then for every degree of freedom gained by user 1 (the weaker user), user 2 (the stronger user)
loses $\min(M,N_2) / \min(M,N_1)$ degrees of freedom. To the best of our
knowledge, this is the first example of an uneven tradeoff in a DoF region.

\subsection{2-User MIMO Interference Channel}
\label{sec:ic_rslts}
In this section, we provide an outer bound and a partial characterization of the DoF region of the
2-user MIMO interference channel.
Please note that although the condition $N_1 \leq N_2$ is seen in the following theorems, the results can
be easily extended for the case $N_1 > N_2$ by symmetry.

\begin{theorem}
\label{thm:dofic1}
{\it (A case where the absence of CSIT does not reduce DoF region)}
If $M_2 \leq N_1 \leq N_2,$ the DoF region of the 2-user MIMO interference channel with no CSIT, as defined in Section \ref{sec:icmodel}, is the same with that of the same channel under the assumption of perfect CSIT \cite{Jafar_Fakhereddin}.
\end{theorem}

\begin{figure}
\begin{center}
\includegraphics[width=400pt, trim=0 100 0 120, clip]{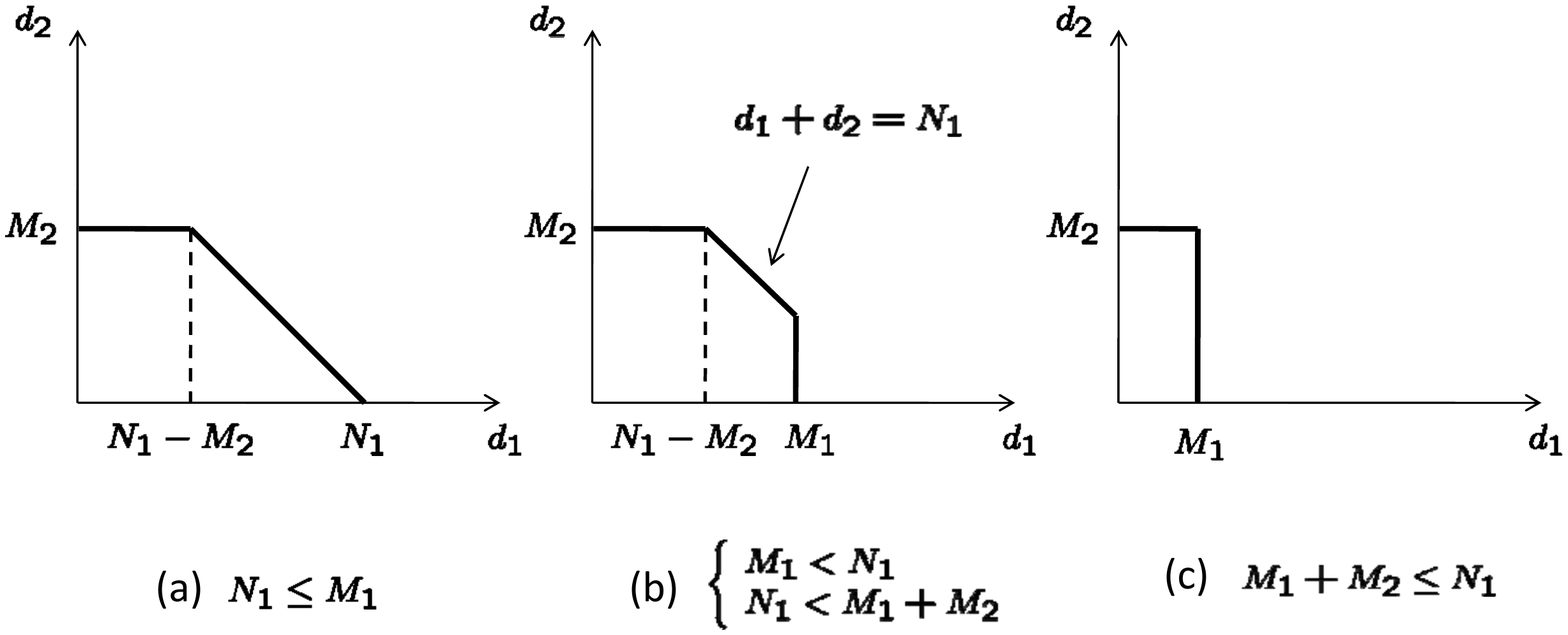}
\caption{The DoF region of the 2-user MIMO interference channel when $M_2 \leq N_1 \leq N_2$.}
\end{center}
\end{figure}

\begin{theorem}
\label{thm:dofic_out}
{\it (A DoF outer bound)}
If $N_1 \leq N_2$ and $N_1 < M_2,$ the DoF region, denoted as ${\bf D}$, of the 2-user MIMO
interference channel with no CSIT, as defined in Section \ref{sec:icmodel}, satisfies
\begin{equation}
{\bf D} \subseteq {\bf D}_{out} \triangleq
\left\{(d_1,d_2)\in\mathbb{R}_2^+: \frac{d_1}{N_1}+\frac{d_2}{\min(M_2,N_2)}\leq 1\right\}.
\label{eq:dofic_out}
\end{equation}
\end{theorem}

\begin{theorem}
\label{thm:dofic2}
{\it (Tightness of the DoF outer bound)}
If $N_1 \leq N_2,$ $N_1 < M_2,$  and $N_1 \leq M_1,$ the DoF outer bound given in (\ref{eq:dofic_out})
is tight, i.e., ${\bf D} = {\bf D}_{out}.$
\end{theorem}

\begin{figure}
\begin{center}
\includegraphics[width=400pt, trim=0 160 0 120, clip]{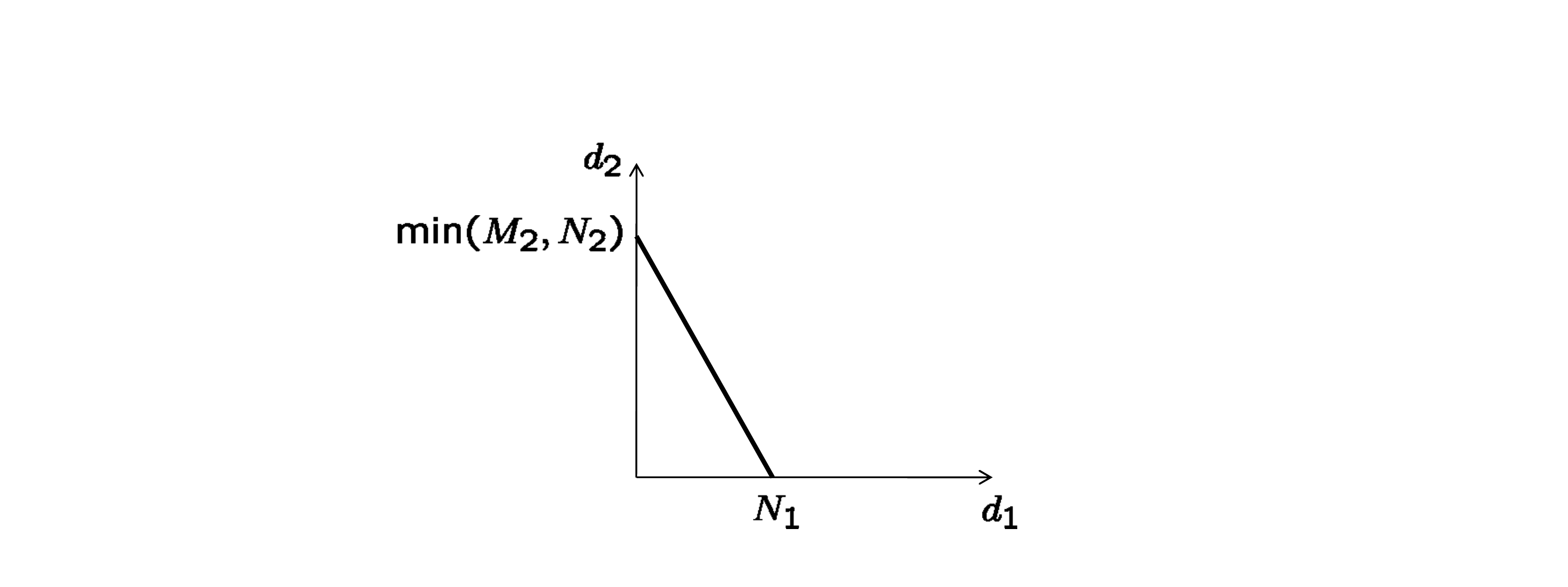}
\caption{The DoF region of the 2-user MIMO interference channel when $N_1 \leq N_2,$ $N_1 < M_2,$ and $N_1 \leq M_1.$}
\end{center}
\end{figure}

We prove these theorems in the Section \ref{sec:prf}.

{\it Remark 1:} To completely characterize of the DoF region of the 2-user MIMO interference
channel, one needs to solve the remaining case that $N_1 < N_2,$ $N_1 < M_2$, and $ M_1 < N_1.$
Figure \ref{fig:f3} shows the gap between ${\bf D}_{out}$ and the best achievable region achieved by time division
multiplexing and receiver zero forcing when $N_1 \leq N_2,$ $N_1 < M_2$, and $ M_1 < N_1.$
For example, for the case $(M_1, M_2, N_1, N_2) = (1, 3, 2, 4)$, while
$(d_1, d_2)=(1, \frac{3}{2}) \in {\bf D}_{out},$
we can only achieve $(d_1, d_2)=(1, 1)$ by zero forcing at receivers.
It's interesting to note that when $N_1 = N_2,$ points A and B in Figure \ref{fig:f3} coincide with each other,
establishing the tightness of the outer bound. Thus, when $N_1 = N_2,$ we are also able to provide a complete
characterization of the DoF region of the channel.

\begin{figure}
\begin{center}
\includegraphics[width=400pt, trim=0 120 0 80, clip]{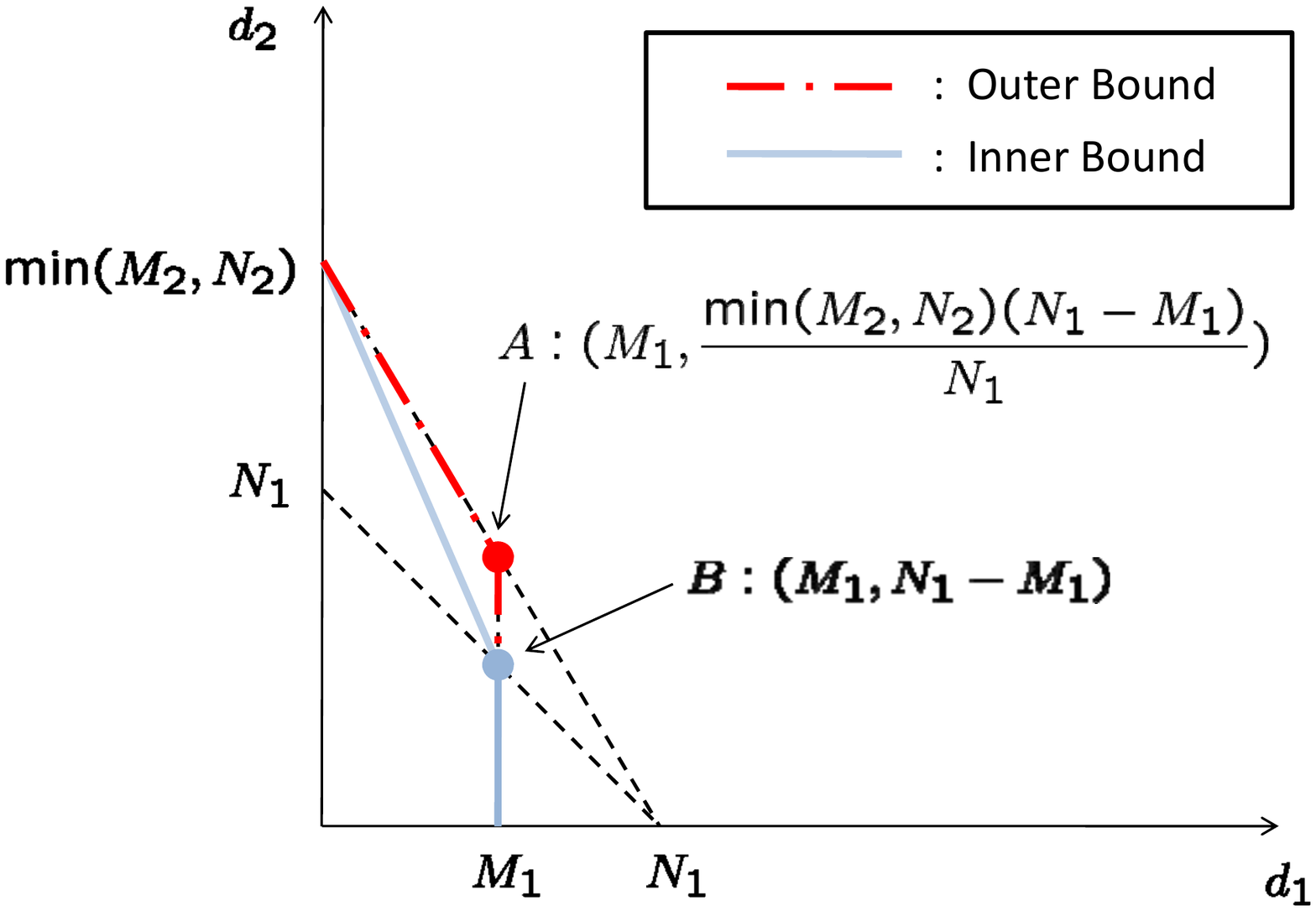}
\caption{The gap between the inner and outer bound of the DoF region when $N_1 \leq N_2,$ $N_1 < M_2$, and $ M_1 < N_1.$} \label{fig:f3}
\end{center}
\end{figure}

\begin{figure}
\begin{center}
\includegraphics[width=400pt, trim=0 100 0 80, clip]{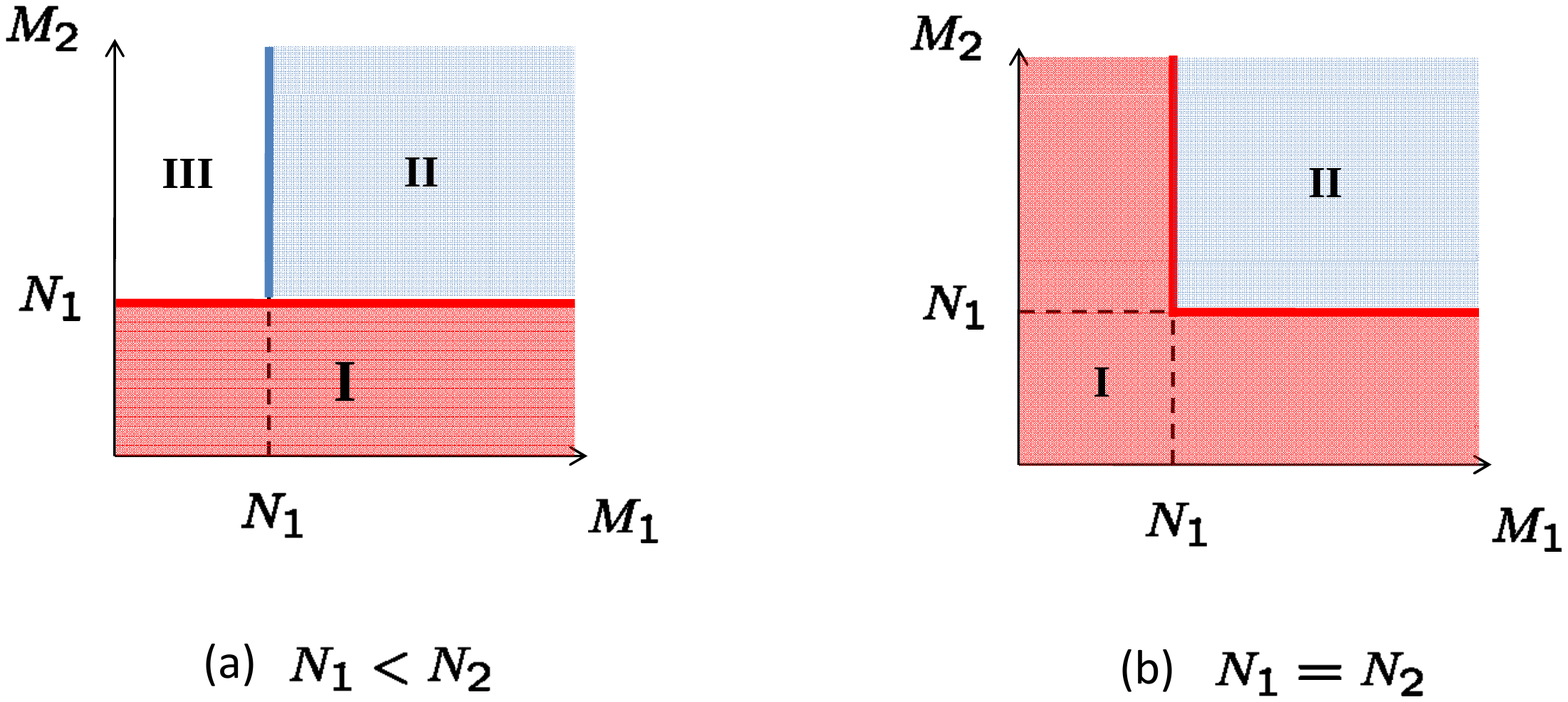}
\caption{The plot showing the relative magnitudes of $(M_1,M_2,N_1,N_2).$} \label{fig:rgm}
\end{center}
\end{figure}

\begin{table}[hbt]
\begin{center}
\caption{DoF region of the 2-user MIMO interference channel when $N_1 < N_2.$}\label{tbl:icdof1}
\begin{tabular}{ | c | c | c | c | c |}
\hline
Case Number & Case Definition & DoF Region & Achievable Scheme & Comparison \\
\hline
I & $M_2 \leq N_1 $ & $ \left\{ \begin{array}{l} d_1 \leq \min(M_1, N_1) \\
d_2 \leq M_2 \\
d_1+d_2 \leq N_1 \end{array} \right. $ & Receiver Zero Forcing & ${\bf D} = {\bf D}_{CSIT}$\\
\hline
II &
$ \left\{ \begin{array}{l} N_1 < M_2 \\
N_1 \leq M_1 \end{array} \right.$ &
$ \frac{d_1}{N_1} + \frac{d_2}{\min(M_2, N_2)} \leq 1$
& Time Division & ${\bf D} \subset {\bf D}_{CSIT}$ \\
\hline
III &
$ \left\{ \begin{array}{l} N_1 < M_2 \\
M_1 < N_1 \end{array} \right.$ &
Unknown
& Unknown & ${\bf D} \subset {\bf D}_{CSIT}$ \\
\hline
\end{tabular}
\end{center}
\end{table}

\begin{table}[hbt]
\begin{center}
\caption{DoF region of the 2-user MIMO interference channel when $N_1 = N_2.$}\label{tbl:icdof2}
\begin{tabular}{ | c | c | c | c | c |}
\hline
Case Number & Case Definition & DoF Region & Achievable Scheme & Comparison \\
\hline
I & $M_2 \leq N_1 $ or $M_1 \leq N_1 $ &
$ \left\{ \begin{array}{l}
d_1 \leq M_1 \\
d_2 \leq M_2 \\
d_1+d_2 \leq N_1 \end{array} \right. $ & Receiver Zero Forcing & ${\bf D} = {\bf D}_{CSIT}$\\
\hline
II &
$ \left\{ \begin{array}{l} N_1 < M_2 \\
N_1 < M_1 \end{array} \right.$ &
$ d_1 + d_2 \leq N_1$
& Time Division & ${\bf D} \subset {\bf D}_{CSIT}$ \\
\hline
\end{tabular}
\end{center}
\end{table}

{\it Remark 2:} Note that, unlike the case in Theorem \ref{thm:dofic1}, when $N_1 < N_2$ and $N_1 < M_2,$
DoF region shrinks (being strictly smaller) due to the absence of CSIT. Specifically, when
$N_1 < N_2$ and $N_1 < M_2,$ the DoF region satisfies
\begin{equation}
{\bf D} \subset {\bf D}_{CSIT},
\end{equation}
where ${\bf D}_{CSIT}$ denotes the DoF region of the same channel under the assumption of perfect CSIT
and CSIR. Moreover, when $N_1 < N_2,$ $N_1 < M_2,$ and $N_1 \leq M_1,$ the degrees of freedom are lost
to the extent that a simple time division multiplexing between the two users is optimal.
These results are summarized in Table \ref{tbl:icdof1}, and the parameter regime of $(M_1, M_2,
N_1, N_2)$ is illustrated in Figure \ref{fig:rgm}.

\subsection{Interference Alignment on the 2 User MIMO Interference channel with no CSIT}
Interference alignment is typically associated with scenarios where one receiver faces interference from distributed transmitters (the compound BC \cite{Weingarten_Shamai_Kramer, Gou_Jafar_Wang} is a notable exception). Moreover, it is often associated with scenarios where the transmitters have perfect channel knowledge. Therefore it is interesting to note that the possibility of interference alignment arises even in the two user MIMO interference channel with no CSIT, where a receiver sees interference from only one transmitter. The following discussion  illustrates this point.
\subsubsection{A Simple Example} The interference alignment scheme is described as follows.
Suppose that $M_1 = N_1 = 1,$ $M_2 = 3,$ and $N_2 = 4.$ The DoF region for this case is a triangle with
vertices $(d_1, d_2)=(0,0),$ $(1,0),$ and $(0,3).$ Now consider the point $(0.5, 1.5)$ on the boundary of
the DoF region. This point can be achieved by an even time division between the operation point
$(1,0)$ and $(0,3),$ which will require each transmitter to use power $P$. Interestingly, it can also be
achieved with simultaneous transmission by both users, such that user 1 sends one beam with power $P$, and user 2 sends three beams, each with power $\sqrt{P}.$ Since user 2 has four receive antennas, it can resolve all four beams and get $0.5 + 0.5 + 0.5 = 1.5$ degrees of freedom. At receiver 1, by treating interference as noise, user 1 get $0.5$ degree of freedom, because the signal power is $P$ and the noise power is
${3 \|{\bf H}^{[12]}\|}^2 \sqrt{P} + 1.$ The interference aligns in the sense that even though there are
three interfering beams, the power of the interference is the same in log scale as if there is only one
beam, i.e., the strongest exponent dominates. Clearly, the interference alignment solution is better than
the time division solution because user 2 only needs power $3 \sqrt{P}$ instead of
scaling power like $P$. Also note that if user 2's transmit power is constrained to scale like
only $\sqrt{P},$ then the time division solution is not feasible while the alignment solution still is optimal.

\subsubsection{The Case where DoF Region Remains Open}
Finally, the possibility of interference alignment is also the remaining hurdle that needs to be overcome for the class of 2-user MIMO  interference channels (class III in Table I) where the DoF region remains unknown. As an illustrative example of this class of channels let us consider again the case $(M_1, M_2,N_1,N_2)=(1,3,2,4)$ and the open problem -- what is the max. number of DoF achieved by user 2 when user 1 achieves 1 DoF ? As stated earlier, the known inner bounds only allow us to achieve 1 DoF for user 2, whereas our outer bound shows that user 2 cannot achieve more than 1.5 DoF in this case. While it is beyond the scope of this work, we expect that if we consider a a finite state compound channel setting, i.e. the channel coefficients are drawn arbitrarily from a discrete set of finite cardinality consisting of generic values -- also known as the finite state compound channel setting -- then following similar arguments as \cite{Gou_Jafar_Wang} it will be possible to show that without CSIT, user 2 can indeed achieve 1.5 DoF (almost surely) at the same time that user 1 achieves 1 DoF. In that case, the open problem bears interesting parallels to the general (not isotropic) MISO broadcast channel with no CSIT where the DoF are unknown with channel uncertainty over a continuous space \cite{Lapidoth_Shamai_Wigger} but the outer bounds are achievable under the finite state compound channel setting \cite{Gou_Jafar_Wang}. Thus, settling the remaining open problem (Case III)  for the DoF region of the 2-user MIMO IC -- especially under general channel uncertainty models -- may be as difficult as settling the corresponding issue for the MISO broadcast channel.

\section{Proof}
\label{sec:prf}

\subsection{Proof of Theorem \ref{thm:dofbc}}
Without loss of generality, let us assume $N_1\leq N_2$.
The case where $M\leq N_1\leq N_2$ is trivial, because in this case the degrees of freedom region, even
with perfect CSIT, is given by:
\begin{equation}
\mathbf{D}=\{(d_1,d_2)\in\mathbb{R}_2^+: d_1+d_2\leq M\},
\end{equation}
which is clearly achievable even without CSIT, by simple time-division between the two users.

For the remainder of this section, we consider $M\geq N_1$.

Since the MIMO broadcast channel without CSIT, as defined in Section \ref{sec:bcmodel}, falls in the class
of degraded broadcast channels \cite{Bergmans,Bergmans1}, its capacity region $C(P)$ is the set of rate
pairs $(R_1,R_2)$ satisfying:
\begin{eqnarray}
R_1
&\leq&
I(U;{\bf Y}^{[1]}| {\bf H}^{[1]},{\bf H}^{[2]})\label{eq:R_1bound}\\
&=&
I({\bf X};{\bf Y}^{[1]}|{\bf H}^{[1]},{\bf H}^{[2]})-I({\bf X};{\bf Y}^{[1]}|{\bf H}^{[1]},{\bf H}^{[2]},U)\\
R_2
&\leq&
I({\bf X};{\bf Y}^{[2]}|{\bf H}^{[1]},{\bf H}^{[2]},U)\label{eq:R_2bound}
\end{eqnarray}
where $U\rightarrow{\bf X}\rightarrow({\bf Y}^{[1]},{\bf Y}^{[2]})$ forms a Markov chain. Since the channel between the transmitter and receiver 1 cannot have more than $\min(M,N_1)=N_1$ degrees of freedom, we have:
\begin{equation}
I({\bf X};{\bf Y}^{[1]}|{\bf H}^{[1]},{\bf H}^{[2]}) \leq N_1\log(P)+o(\log(P)).
\end{equation}
Let us define $r$ as the prelog for the term $I({\bf X};{\bf Y}^{[1]}|{\bf H}^{[1]},{\bf H}^{[2]},U).$ Thus, we have
\begin{eqnarray}
r\log(P)+o(\log(P))
& = &
I({\bf X};{\bf Y}^{[1]}|{\bf H}^{[1]},{\bf H}^{[2]},U) \label{eq:definer} \\
& \stackrel{(a)}{=} &
\sum_{i=1}^{N_1}I({\bf X};{\bf Y}^{[1]}_i|{\bf H}^{[1]},{\bf H}^{[2]},U,{\bf Y}^{[1]}_{i+1:N_1}) \\
& \stackrel{(b)}{\geq} &
N_1 I({\bf X};{\bf Y}^{[1]}_1|{\bf H}^{[1]},{\bf H}^{[2]},U,{\bf Y}^{[1]}_{2:N_1}), \label{eq:preuseful}
\end{eqnarray}
where $(a)$ follows by chain rule and $(b)$ follows by the fact that adding conditioning does not increase the differential entropy. Dividing both sides by $N_1,$ we have
\begin{equation}
I({\bf X};{\bf Y}^{[1]}_1|{\bf H}^{[1]},{\bf H}^{[2]},U,{\bf Y}^{[1]}_{2:N_1})
\leq
\frac{r}{N_1}\log(P)+o(\log(P)). \label{eq:useful}
\end{equation}
Now we can write the upperbound (\ref{eq:R_2bound}) for $R_2$ as follows.
\begin{eqnarray}
R_2
& \leq &
I({\bf X}; {\bf Y}^{[2]}|{\bf H}^{[1]},{\bf H}^{[2]},U) \\
& = &
I({\bf X}; {\bf Y}^{[2]}_{1:\min(M,N_2)}, {\bf Y}^{[2]}_{\min(M,N_2)+1:N_2} | {\bf H}^{[1]},{\bf H}^{[2]},U) \\
& = &
I({\bf X}; {\bf Y}^{[2]}_{1:\min(M,N_2)} | {\bf H}^{[1]},{\bf H}^{[2]},U) +
\underbrace{I({\bf X};{\bf Y}^{[2]}_{\min(M,N_2)+1:N_2} | {\bf H}^{[1]},{\bf H}^{[2]},U,{\bf Y}^{[2]}_{1:\min(M,N_2)})}_{o(\log(P))} \label{eq:redundant} \\
& = &
I({\bf X}; {\bf Y}^{[2]}_{1:N_1}, {\bf Y}^{[2]}_{N_1+1:\min(M,N_2)} | {\bf H}^{[1]},{\bf H}^{[2]},U)+o(\log(P)) \\
& \stackrel{(a)}{=} &
I({\bf X}; {\bf Y}^{[1]}, {\bf Y}^{[2]}_{N_1+1:\min(M,N_2)} | {\bf H}^{[1]},{\bf H}^{[2]},U)+o(\log(P)) \\
& = &
I({\bf X}; {\bf Y}^{[1]} | {\bf H}^{[1]},{\bf H}^{[2]},U) +
I({\bf X}; {\bf Y}^{[2]}_{N_1+1:\min(M,N_2)} | {\bf H}^{[1]}, {\bf H}^{[2]}, U, {\bf Y}^{[1]})+o(\log(P)), \label{eq:bc_R2}
\end{eqnarray}
where $(a)$ follows by the fact that all channels are statistically equivalent.
The second term in (\ref{eq:bc_R2}) can be upperbounded as follows.
\begin{eqnarray}
I({\bf X}; {\bf Y}^{[2]}_{N_1+1:\min(M,N_2)} | {\bf H}^{[1]}, {\bf H}^{[2]}, U, {\bf Y}^{[1]})
& \stackrel{(a)}{=} &
\sum_{i=N_1+1}^{\min(M,N_2)} I({\bf X}; {\bf Y}^{[2]}_i | {\bf H}^{[1]}, {\bf H}^{[2]}, U, {\bf Y}^{[1]}, {\bf Y}^{[2]}_{i+1:\min(M,N_2)}) \\
& \stackrel{(b)}{\leq} &
\sum_{i=N_1+1}^{\min(M,N_2)} I({\bf X}; {\bf Y}^{[2]}_i | {\bf H}^{[1]}, {\bf H}^{[2]}, U, {\bf Y}^{[1]}_{2:N_1}) \\
& \stackrel{(c)}{\leq} &
\sum_{i=N_1+1}^{\min(M,N_2)} I({\bf X}; {\bf Y}^{[2]}_1 | {\bf H}^{[1]}, {\bf H}^{[2]}, U, {\bf Y}^{[1]}_{2:N_1}) \\
& \stackrel{(d)}{=} &
(\min(M, N_2) - N_1) \frac{r}{N_1} \log(P) + o(\log(P)), \label{eq:bc_R2_T2}
\end{eqnarray}
where $(a)$ follows by chain rule, $(b)$ follows by the fact that dropping conditioning does not
increase differential entropy, $(c)$ follows by the fact that all channels are statistically equivalent,
and $(d)$ follows (\ref{eq:useful}).
Substituting (\ref{eq:definer}) and (\ref{eq:bc_R2_T2}) into (\ref{eq:bc_R2}), we have
\begin{equation}
R_2 \leq \frac{\min(M,N_2)}{N_1}r\log(P)+o(\log(P)).
\end{equation}

Thus, for $0\leq r\leq N_1$, an outer bound on the boundary of the degrees of freedom region is characterized as follows.
\begin{eqnarray}
(d_1,d_2)=\left(N_1-r,  \frac{\min(M,N_2)}{N_1}r\right),
\end{eqnarray}
which implies that
\begin{eqnarray}
\mathbf{D}\subset \left\{(d_1,d_2)\in\mathbb{R}_2^+: \frac{d_1}{N_1}+\frac{d_2}{\min(M,N_2)}\leq 1\right\}.
\end{eqnarray}
Achievability of this outer bound follows trivially by time division between the two users.
This concludes the proof.

\subsection{Proof of Theorem \ref{thm:dofic1}}
Let a genie provide the transmitters with perfect channel state information.
Since giving CSIT does not hurt, the converse argument is still valid. Using the DoF result
of the 2-user MIMO interference channel with perfect CSIT \cite{Jafar_Fakhereddin}, we have
\begin{equation}
d_1 + d_2
\leq
\min \left( \max(M_1, N_2), \max(M_2, N_1) \right)
\leq
N_1.
\end{equation}
Achievability of this outer bound follows trivially by receiver zero forcing. This concludes the proof.

\subsection{Proof of Theorem \ref{thm:dofic_out}}
Before providing the proof, we would like to mention that this proof follows similar
lines with those in the proof of Theorem \ref{thm:dofbc}. We include the detailed proof
to show the subtle, though important, differences and for the sake of completeness.
For brevity, we use ${\bf H}$ to denote the set of all channel matrices. That is
\begin{equation}
{\bf H} = \{ {\bf H}^{[11]}, {\bf H}^{[12]}, {\bf H}^{[21]}, {\bf H}^{[22]} \}.
\end{equation}
For any codeword with $n$ channel usages, using Fano's inequality, we have
\begin{eqnarray}
n(R_1 - \epsilon_n)
&\leq&
I(W_1; {{\bf Y}^{[1]}}^n | {\bf H}) \\
&=&
I(W_1 W_2; {{\bf Y}^{[1]}}^n | {\bf H}) - I(W_2; {{\bf Y}^{[1]}}^n | W_1, {\bf H}) \\
& \stackrel{(a)}{\leq} &
nN_1 \log(P) + o(\log(P)) - I(W_2; {{\bf Y}^{[1]}}^n | W_1, {\bf H}), \label{eq:ic_nR1}
\end{eqnarray}
where (a) follows from the fact that the degrees of freedom of a 2-user MIMO multiple access channel is bounded by the number of receive antennas.

Let us define $r$ as the prelog of the term $\frac{1}{n} I(W_2; {{\bf Y}^{[1]}}^n | W_1, {\bf H}).$ Thus, we have
\begin{eqnarray}
n r \log(P) + o(\log(P))
&=&
I(W_2; {{\bf Y}^{[1]}}^n | W_1, {\bf H}) \label{eq:ic_I_r}\\
& \stackrel{(a)}{=} &
\sum_{i=1}^{N_1} I(W_2; {{\bf Y}^{[1]}_i}^n | W_1, {\bf H}, {\bf Y}^{[1]^{\scriptstyle{n}}}_{i+1:N_1}) \\
& \stackrel{(b)}{=} &
\sum_{i=1}^{N_1} I(W_2; {{\bf Y}^{[1]}_1}^n | W_1, {\bf H}, {\bf Y}^{[1]^{\scriptstyle{n}}}_{i+1:N_1}) \\
& \stackrel{(c)}{\geq} &
N_1 I(W_2; {{\bf Y}^{[1]}_1}^n | W_1, {\bf H}, {\bf Y}^{[1]^{\scriptstyle{n}}}_{2:N_1}),
\end{eqnarray}
where (a) follows by chain rule, (b) follows by the fact that all channels are statistically equivalent, and (c) follows by the fact that adding conditioning does not increase differential entropy. Dividing both sides by $N_1,$ we have
\begin{equation}
I(W_2; {{\bf Y}^{[1]}_1}^n | W_1, {\bf H}, {\bf Y}^{[1]^{\scriptstyle{n}}}_{2:N_1})
\leq
\frac{nr}{N_1} \log(P) + o(\log(P)). \label{eq:ic_per_antenna}
\end{equation}

Now we can upperbound $R_2$ as follows.
\begin{eqnarray}
n(R_2 - \epsilon_n)
& \stackrel{(a)}{\leq} &
I(W_2;{{\bf Y}^{[2]}}^n | {\bf H}) \\
& \leq &
I(W_2; {{\bf Y}^{[2]}}^n, W_1 | {\bf H}) \\
& = &
I(W_2; {{\bf Y}^{[2]}}^n | {\bf H}, W_1) \\
& = &
I(W_2; {\bf Y}^{[2]^{\scriptstyle n}}_{1:\min(M_2, N_2)}, {\bf Y}^{[2]^{\scriptstyle n}}_{\min(M_2, N_2)+1:N_2} | {\bf H}, W_1) \\
& = &
I(W_2; {\bf Y}^{[2]^{\scriptstyle n}}_{1:\min(M_2, N_2)} | {\bf H}, W_1) +
\underbrace{I(W_2; {\bf Y}^{[2]^{\scriptstyle n}}_{\min(M_2, N_2)+1:N_2} | {\bf H}, W_1, {\bf Y}^{[2]^{\scriptstyle n}}_{1:\min(M_2, N_2)})}_{o(\log(P))} \\
& = &
I(W_2; {\bf Y}^{[2]^{\scriptstyle n}}_{1:N_1}, {\bf Y}^{[2]^{\scriptstyle n}}_{N_1+1:\min(M_2, N_2)} | {\bf H}, W_1) + o(\log(P)) \\
& \stackrel{(b)}{=} &
I(W_2; {\bf Y}^{[1]^{\scriptstyle n}}, {\bf Y}^{[2]^{\scriptstyle n}}_{N_1+1:\min(M_2, N_2)} | {\bf H}, W_1) + o(\log(P)) \\
& = &
I(W_2; {\bf Y}^{[1]^{\scriptstyle n}} | {\bf H}, W_1) +
I(W_2; {\bf Y}^{[2]^{\scriptstyle n}}_{N_1+1:\min(M_2, N_2)} | {\bf H}, W_1, {\bf Y}^{[1]^{\scriptstyle n}}) + o(\log(P)), \label{eq:ic_nR2}
\end{eqnarray}
where $(a)$ follows by Fano's inequality and $(b)$ follows by the fact that
all channels are statistically equivalent.
The second term in (\ref{eq:ic_nR2}) can be upperbounded as follows.
\begin{eqnarray}
I(W_2; {\bf Y}^{[2]^{\scriptstyle n}}_{N_1+1:\min(M_2, N_2)} | {\bf H}, W_1, {\bf Y}^{[1]^{\scriptstyle n}}) & \stackrel{(a)}{\leq} &
I(W_2; {\bf Y}^{[2]^{\scriptstyle n}}_{N_1+1:\min(M_2, N_2)} | {\bf H}, W_1, {\bf Y}^{[1]^{\scriptstyle n}}_{2:N_1}) \\
& = &
\sum_{i=N_1+1}^{\min(M_2,N_2)}
I(W_2; {\bf Y}^{[2]^{\scriptstyle n}}_i | {\bf H}, W_1, {\bf Y}^{[1]^{\scriptstyle n}}_{2:N_1}, {\bf Y}^{[2]^{\scriptstyle n}}_{i+1:\min(M_2, N_2)}) \\
& \stackrel{(b)}{\leq} &
\sum_{i=N_1+1}^{\min(M_2,N_2)}
I(W_2; {\bf Y}^{[2]^{\scriptstyle n}}_i | {\bf H}, W_1, {\bf Y}^{[1]^{\scriptstyle n}}_{2:N_1}) \\
& \stackrel{(c)}{=} &
\sum_{i=N_1+1}^{\min(M_2,N_2)}
I(W_2; {\bf Y}^{[1]^{\scriptstyle n}}_1 | {\bf H}, W_1, {\bf Y}^{[1]^{\scriptstyle n}}_{2:N_1}) \\
& \stackrel{(d)}{\leq} &
(\min(M_2,N_2)- N_1) \frac{nr}{N_1} \log(P) + o(\log(P)) \label{eq:ic_nR2_T2},
\end{eqnarray}
where $(a)$ and $(b)$ follow by the fact that dropping conditioning does not increase differential entropy, $(c)$ follows by the fact that all channels are statistically equivalent, and $(d)$ follows by (\ref{eq:ic_per_antenna}).

Substituting (\ref{eq:ic_I_r}) and (\ref{eq:ic_nR2_T2}) into (\ref{eq:ic_nR2}), substituting (\ref{eq:ic_I_r}) into (\ref{eq:ic_nR1}), and letting $n$ goes to infinity, we have
\begin{eqnarray}
R_1 &\leq& (N_1 - r) \log(P) + o(\log(P)) \\
R_2 &\leq& \frac{\min(M_2, N_2)}{N_1} r \log(P) + o(\log(P)),
\end{eqnarray}
where
\begin{equation}
0 \leq r \leq N_1.
\end{equation}
This concludes the proof.

\subsection{Proof of Theorem \ref{thm:dofic2}}
The achievability for this case follows trivially by time division between the two users.

\section{Capacity Region of a Class of Broadcast Channels with No CSIT}
\label{sec:cpcty_bc}
\subsection{Models}
\label{subsec:bcmodel2}
Consider the $2$-user Gaussian MIMO broadcast channel where the transmitter is equipped with $M$ antennas and receivers $1,2$ are equipped with $N_1,N_2$ antennas, respectively. $M,$ $N_1,$ and $N_2$ are assumed to satisfy
\begin{eqnarray}
N_1 &\leq& M \\
N_2 &\leq& M.
\end{eqnarray}
The channel is described by the input-output relationship:
\begin{eqnarray}
{\bf Y}^{[1]}(t)&=&{\bf H}^{[1]}{\bf Q}(t){\bf X}(t)+{\bf Z}^{[1]}(t)\\
{\bf Y}^{[2]}(t)&=&{\bf H}^{[2]}{\bf Q}(t){\bf X}(t)+{\bf Z}^{[2]}(t)
\end{eqnarray}
where the notation usage for ${\bf X}^{[i]}(t)$ and ${\bf Y}^{[i]}(t)$, the assumption for the noise term ${\bf Z}^{[i]}(t),$ and the assumption that the channel is equipped with perfect CSIR and no CSIT are the same with those in Section \ref{sec:systemmodel}.
However, different from the previous assumption, ${\bf H^{[i]}}$ is assumed to be a time-invariant $N_i\times M$ channel matrix with $N_i$ orthonormal rows, $i\in\{1,2\}.$ Note that this is possible only when $N_1 \leq M$ and $N_2 \leq M$.
${\bf Q}$ is an $M \times M$ matrix whose elements are independent identically distributed circularly symmetric complex Gaussian random variables with zero mean and unit variance, implying that
\begin{eqnarray}
\mbox{E} [{\bf QQ}^\dag] = M {\bf I},
\end{eqnarray}
where {\bf I} is a $M \times M$ identity matrix. The transmit power constraint and the standard definition of the capacity region are the same with those in Section \ref{sec:systemmodel} and we omit them for brevity.

\subsection{Main Result}
\begin{theorem}\label{theorem:capcacitybc}
The capacity region of the MIMO BC with no CSIT, as defined in Section \ref{subsec:bcmodel2} is the following:
\begin{eqnarray}
{\bf C}=\left\{(R_1,R_2)\in\mathbb{R}_2^+: \frac{R_1}{N_1}+\frac{R_2}{N_2}\leq \log(1+P)\right\}.
\end{eqnarray}
\end{theorem}
\begin{proof}
The proof follows the similar lines in the proof of Theorem \ref{thm:dofbc} and we omit the parts that are the same with the those given in Section \ref{sec:prf} for brevity.
Without loss of generality, let us assume $N_1\leq N_2$. Following (\ref{eq:R_1bound}), we have
\begin{equation}
R_1
\leq
I({\bf X};{\bf Y}^{[1]}|{\bf H}^{[1]},{\bf H}^{[2]})-I({\bf X};{\bf Y}^{[1]}|{\bf H}^{[1]},{\bf H}^{[2]},U), \label{eq:R_1bound2}
\end{equation}
where $U\rightarrow{\bf X}\rightarrow({\bf Y}^{[1]},{\bf Y}^{[2]})$ forms a Markov chain. Denote the capacity of the point-to-point link from the transmitter to receiver $1$ as $C_1$ and let
\begin{eqnarray}
\gamma = I({\bf X};{\bf Y}^{[1]}|{\bf H}^{[1]},{\bf H}^{[2]},U).
\end{eqnarray}
Following (\ref{eq:R_1bound2}), we have
\begin{eqnarray}
R_1
&\leq&
C_1 - \gamma \nonumber \\
&=&
\mbox{E}_{\bf Q} \log \left| {\bf I} + {\bf H}^{[1]} {\bf Q} ({\frac{\scriptstyle P}{\scriptstyle M}} {\bf{I}}) {\bf Q}^\dag {{\bf H}^{[1]}}^\dag \right| - \gamma \nonumber \\
&=&
\log \left| {\bf I} + \frac{\scriptstyle P}{\scriptstyle M} {\bf H}^{[1]} \mbox{E}[{\bf QQ}^\dag] {{\bf H}^{[1]}}^\dag \right| - \gamma \nonumber \\
&=&
\log \left| {\bf I} + \frac{\scriptstyle P}{\scriptstyle M} {\bf H}^{[1]} (M{\bf I}) {{\bf H}^{[1]}}^\dag \right| - \gamma \nonumber \\
&=&
\log \left| {\bf I} + P {\bf H}^{[1]} {{\bf H}^{[1]}}^\dag \right| - \gamma \nonumber \\
&=&
N_1 \log(1+P) - \gamma.
\end{eqnarray}
Following (\ref{eq:preuseful}), we have following useful inequality:
\begin{equation}
\frac{1}{N_1}\ I({\bf X};{\bf Y}^{[1]}|{\bf H}^{[1]},{\bf H}^{[2]},U)
\geq
I({\bf X};{\bf Y}^{[1]}_1|{\bf H}^{[1]},{\bf H}^{[2]},U,{\bf Y}^{[1]}_{2:N_1}).\label{eq:usefulcapcity}
\end{equation}

Now, following (\ref{eq:R_2bound}), we can write the upperbound for $R_2$ as follows.
\begin{eqnarray}
R_2
& \leq &
I({\bf X}; {\bf Y}^{[2]} | {\bf H}^{[1]}, {\bf H}^{[2]},U) \\
& = &
I({\bf X}; {\bf Y}^{[2]}_{1:N_1}, {\bf Y}^{[2]}_{N_1+1:N_2} | {\bf H}^{[1]}, {\bf H}^{[2]},U) \\
& = &
I({\bf X}; {\bf Y}^{[1]}, {\bf Y}^{[2]}_{N_1+1:N_2} | {\bf H}^{[1]}, {\bf H}^{[2]},U) \\
& = &
I({\bf X}; {\bf Y}^{[1]} | {\bf H}^{[1]}, {\bf H}^{[2]},U) +
I({\bf X}; {\bf Y}^{[2]}_{N_1+1:N_2} | {\bf H}^{[1]}, {\bf H}^{[2]}, U, {\bf Y}^{[1]}) \\
& = &
\gamma +
\sum_{i=N_1+1}^{N_2} I({\bf X}; {\bf Y}^{[2]}_i | {\bf H}^{[1]}, {\bf H}^{[2]}, U, {\bf Y}^{[1]},
{\bf Y}^{[2]}_{i+1:N_2}) \\
& \leq &
\gamma + \sum_{i=N_1+1}^{N_2} I({\bf X}; {\bf Y}^{[2]}_i | {\bf H}^{[1]}, {\bf H}^{[2]}, U,
{\bf Y}^{[1]}_{2:N_1}) \\
& \leq &
\gamma + \frac{N_2 - N_1}{N_1} \gamma \\
& \leq &
\frac{N_2}{N_1} \gamma.
\end{eqnarray}
Thus, for $0\leq \gamma \leq N_1 \log(1+P)$, an outer bound on the boundary of the capacity region is characterized as follows.
\begin{eqnarray}
(R_1,R_2)=\left(N_1 \log(1+P) - \gamma,  \frac{N_2}{N_1} \gamma\right),
\end{eqnarray}
which implies that
\begin{eqnarray}
\mathbf{C}\subset\left\{(R_1,R_2)\in\mathbb{R}_2^+: \frac{d_1}{N_1}+\frac{d_2}{N_2}\leq \log(1+P)\right\}.
\end{eqnarray}
To provide the achievability of this outer bound, we first prove that $(0, N_2 \log(1+P))$ is achievable. Denote the capacity of the point-to-point link between the transmitter and receiver $2$ as $C_2,$ and we have
\begin{eqnarray}
C_2
&=&
\mbox{E}_{\bf Q} \log \left| {\bf I} + {\bf H}^{[2]} {\bf Q} ({\frac{\scriptstyle P}{\scriptstyle M}} {\bf{I}}) {\bf Q}^\dag {{\bf H}^{[2]}}^\dag \right| \nonumber \\
&=&
\log \left| {\bf I} + \frac{\scriptstyle P}{\scriptstyle M} {\bf H}^{[2]} \mbox{E}[{\bf QQ}^\dag] {{\bf H}^{[2]}}^\dag \right| \nonumber \\
&=&
\log \left| {\bf I} + \frac{\scriptstyle P}{\scriptstyle M} {\bf H}^{[2]} (M {\bf I}) {{\bf H}^{[2]}}^\dag \right| \nonumber \\
&=&
\log \left| {\bf I} + P {\bf H}^{[2]} {{\bf H}^{[2]}}^\dag \right| \nonumber \\
&=&
N_2 \log(1+P).
\end{eqnarray}
Thus, $(0, N_2 \log(1+P))$ is achievable.
Using time division between the two users, and the proof is complete.
\end{proof}

\section{Summary}
\label{sec:cnlsn}
In this paper, we explore the effect of the absence of channel state information for MIMO networks.
Throughout the paper, we assume perfect CSIR and no CSIT.
We provide the characterization of the DoF region of a 2-user MIMO broadcast channel.
We then provide a partial characterization of the DoF region of a 2-user MIMO interference channel.
The condition where the absence of channel state information results in the shrinkage of the DoF
region is identified.
We also extend the outer bound of the DoF region to find the capacity region for a specific 2-user
MIMO broadcast channel.


\bibliography{Thesis}

\end{document}